\documentclass{IEEEtran}

\usepackage{amsmath,amssymb,cite,multirow,verbatim,url,graphicx}
\usepackage[algo2e]{algorithm2e}

\newtheorem{theorem}{Theorem}

\newtheorem{algorithm}{Algorithm}
\newcommand{\GF}{\mathrm{GF}}

\title{Complexity Analysis of Reed--Solomon Decoding over $\GF(2^m)$ Without Using Syndromes}
\author{
	\authorblockN{Ning Chen and Zhiyuan Yan}\\
	\authorblockA{Department of Electrical and Computer Engineering\\
		Lehigh University, Bethlehem, Pennsylvania 18015, USA\\
		E-mail: \{nic6, yan\}@lehigh.edu}
	\thanks{The material in this paper was presented in part at the IEEE Workshop on Signal Processing Systems, Shanghai, China, October 2007.}
}

\begin{document}
\maketitle

\begin{abstract}
There has been renewed interest in decoding Reed--Solomon (RS) codes without using syndromes recently.
In this paper, we investigate the complexity of syndromeless decoding, and compare it to that of syndrome-based decoding.
Aiming to provide guidelines to practical applications, our complexity analysis focuses on RS codes over characteristic-2 fields, for which some \emph{multiplicative} FFT techniques are not applicable.
Due to moderate block lengths of RS codes in practice, our analysis is complete, without big $O$ notation.
In addition to fast implementation using \emph{additive} FFT techniques, we also consider direct implementation, which is still relevant for RS codes with moderate lengths.
For high rate RS codes, when compared to syndrome-based decoding algorithms, not only syndromeless decoding algorithms require more field operations regardless of implementation, but also decoder architectures based on their direct implementations have higher hardware costs and lower throughput.
We also derive tighter bounds on the complexities of fast polynomial multiplications based on Cantor's approach and the fast extended Euclidean algorithm.
\end{abstract}

\begin{keywords}
Reed--Solomon codes, Decoding, Complexity theory, Galois fields, Discrete Fourier transforms, Polynomials
\end{keywords}

\section{Introduction}
Reed--Solomon (RS) codes are among the most widely used error control codes, with applications in space communications, wireless communications, and consumer electronics~\cite{Wicker94}.
As such, efficient decoding of RS codes is of great interest.
The majority of the applications of RS codes use syndrome-based decoding algorithms such as the Berlekamp--Massey algorithm (BMA)~\cite{Berlekamp68} or the extended Euclidean algorithm (EEA)~\cite{Sugiyama75}.
Alternative hard decision decoding methods  for RS codes without using syndromes were considered in~\cite{Shiozaki88,Shiozaki90,Gao03}.
As pointed out in~\cite{Fedorenko05,Fedorenko06a}, these algorithms belong to the class of frequency-domain algorithms and are related to the Welch--Berlekamp algorithm~\cite{Welch83}.
In contrast to syndrome-based decoding algorithms, these algorithms do not compute syndromes and avoid the Chien search and Forney's formula.
Clearly, this difference leads to the question whether these algorithms offer lower complexity than syndrome-based decoding, especially when fast Fourier transform (FFT) techniques are applied~\cite{Gao03}.

Asymptotic complexity of syndromeless decoding was analyzed in~\cite{Gao03}, and in~\cite{Fedorenko05} it was concluded that syndromeless decoding has the same asymptotic complexity $O(n \log^2 n )$\footnote{Note that all the logarithms in this paper are to base two.} as syndrome-based decoding~\cite{Justesen76}.
However, existing asymptotic complexity analysis is limited in several aspects.
For example, for RS codes over Fermat fields $\GF(2^{2^r}+1)$ and other prime fields~\cite{Shiozaki90,Gao03}, efficient multiplicative FFT techniques lead to an asymptotic complexity of $O(n \log^2 n )$.
However, such FFT techniques do not apply to characteristic-2 fields, and hence this complexity is not applicable to RS codes over characteristic-2 fields.
For RS codes over arbitrary fields, the asymptotic complexity of syndromeless decoding based on multiplicative FFT techniques was shown to be $O(n \log^2 n \log\log n)$~\cite{Gao03}.
Although they are applicable to RS codes over characteristic-2 fields, the complexity has large coefficients and multiplicative FFT techniques are less efficient than fast implementation based on additive FFT for RS codes with moderate block lengths~\cite{Gao03,Gathen96a,Gathen03}.
As such, asymptotic complexity analysis provides little help to practical applications.

In this paper, we analyze the complexity of syndromeless decoding and compare it to that of syndrome-based decoding.
Aiming to provide guidelines to system designers, we focus on the decoding complexity of RS codes over $\GF(2^m)$.
Since RS codes in practice have moderate lengths, our complexity analysis provides not only the coefficients for the most significant terms, but also the following terms.
Due to their moderate lengths, our comparison is based on two types of implementations of syndromeless decoding and syndrome-based decoding: direct implementation and fast implementation based on FFT techniques.
Direct implementations are often efficient when decoding RS codes with moderate lengths and have widespread applications; thus, we consider both computational complexities, in terms of field operations, and hardware costs and throughputs.
For fast implementations, we consider their computational complexities only and their hardware implementations are beyond the scope of this paper.
We use \emph{additive} FFT techniques based on Cantor's approach~\cite{Cantor89} since this approach achieves small coefficients~\cite{Gao03,Gathen96a} and hence is more suitable for moderate lengths.
In contrast to some previous works~\cite{Gathen03,Khodadad05}, which count field multiplications and additions together, we differentiate the multiplicative and additive complexities in our analysis.

The main contributions of the papers are:
\begin{itemize}
	\item We derived a tighter bound on the complexities of fast polynomial multiplication based on Cantor's approach;
	\item We also obtained a tighter bound on the complexity of the fast extended Euclidean algorithm (FEEA) for general partial greatest common divisor (GCD) computation;
	\item We evaluated the complexities of syndromeless decoding based on different implementation approaches and compare them with their counterparts of syndrome-based decoding; Both errors-only and errors-and-erasures decoding are considered.
	\item We compare the hardware costs and throughputs of direct implementations for syndromeless decoders with those for syndrome-based decoders.
\end{itemize}

The rest of the paper is organized as follows.
To make this paper self-contained, in Section~\ref{sec:rev} we briefly review FFT algorithms over finite fields, fast algorithms for polynomial multiplication and division over $\GF(2^m)$, the FEEA, and syndromeless decoding algorithms.
Section~\ref{sec:time} presents both computational complexity and decoder architectures of direct implementations of syndromeless decoding, and compare them with their counterparts for syndrome-based decoding algorithms.
Section~\ref{sec:tran} compares the computational complexity of fast implementations of syndromeless decoding with that of syndrome-based decoding.
In Section~\ref{sec:case}, case studies on two RS codes are provided and errors-and-erasures decoding is discussed.
The conclusions are given in Section~\ref{sec:con}.

\section{Background}\label{sec:rev}

\subsection{Fast Fourier Transform Over Finite Fields}
For any $n$ ($n \mid q-1$) distinct elements $a_0, a_1, \dots, a_{n-1} \in \mathrm{GF}(q)$, the transform from $\boldsymbol{f} = (f_0, f_1, \dots, f_{n-1})^T$ to $\boldsymbol{F} \triangleq \bigl(f(a_0), f(a_1), \dots, f(a_{n-1})\bigr)^T$, where $f(x) = \sum_{i=0}^{n-1}f_i x^i \in \mathrm{GF}(q)[x]$, is called a discrete Fourier transform (DFT), denoted by $\boldsymbol{F}=\mathrm{DFT}(\boldsymbol{f})$.
Accordingly, $\boldsymbol{f}$ is called the inverse DFT of $\boldsymbol{F}$, denoted by $\boldsymbol{f}=\mathrm{IDFT}(\boldsymbol{F})$.
Asymptotically fast Fourier transform (FFT) algorithm over $\mathrm{GF}(2^m)$ was proposed in~\cite{Wang88}.
Reduced-complexity cyclotomic FFT (CFFT) was shown to be efficient for moderate lengths in~\cite{Chen07}.

\subsection{Polynomial Multiplication Over $\GF(2^m)$ By Cantor's Approach}
A fast polynomial multiplication algorithm using additive FFT was proposed by Cantor~\cite{Cantor89} for $\GF(q^{q^m})$, where $q$ is prime, and it was generalized to $\GF(q^m)$ in~\cite{Gathen96a}.
Instead of evaluating and interpolating over the multiplicative subgroups as in multiplicative FFT techniques, Cantor's approach uses additive subgroups.
Cantor's approach relies on two algorithms: multipoint evaluation (MPE)~\cite[Algorithm 3.1]{Gathen96a} and multipoint interpolation (MPI)~\cite[Algorithm 3.2]{Gathen96a}.

Suppose the degree of the product of two polynomials over $\GF(2^m)$ is less than $h$ ($h\le2^m$), the product can be obtained as follows: First, the two operand polynomials are evaluated using the MPE algorithm; The evaluation results are then multiplied point-wise; Finally the product polynomial is obtained by the MPI algorithm to interpolate the point-wise multiplication results.
The polynomial multiplication requires at most $\frac{3}{2}h\log^2h +\frac{15}{2}h\log h+8h$ multiplications over $\GF(2^m)$ and $\frac{3}{2}h\log^2h +\frac{29}{2}h\log h+4h+9$ additions over $\GF(2^m)$~\cite{Gathen96a}.
For simplicity, henceforth in this paper, all arithmetic operations are over $\GF(2^m)$ unless specified otherwise.

\subsection{Polynomial Division By Newton Iteration}
Suppose $a,b \in \mathrm{GF}(q)[x]$ are two polynomials of degrees $d_0+d_1$ and $d_1$ $(d_0, d_1 \ge 0)$, respectively.
To find the quotient polynomial $q$ and the remainder polynomial $r$ satisfying $a = qb + r$ where $\deg r < d_1$, a fast polynomial division algorithm is available~\cite{Gathen03}.
Suppose $\mathrm{rev}_h(a) \triangleq x^h a(\frac{1}{x})$, the fast algorithm first computes the inverse of $\mathrm{rev}_{d_1}(b) \bmod x^{d_0+1}$ by Newton iteration.
Then the reverse quotient is given by $q^*=\mathrm{rev}_{d_0+d_1}(a)\mathrm{rev}_{d_1}(b)^{-1} \bmod x^{d_0+1}$.
Finally, the actual quotient and remainder are given by $q=\mathrm{rev}_{d_0}(q^*)$ and $r=a-qb$.

Thus, the complexity of polynomial division with remainder of a polynomial $a$ of degree $d_0+d_1$ by a monic polynomial $b$ of degree $d_1$ is at most $4\mathsf{M}(d_0) + \mathsf{M}(d_1) + O(d_1)$ multiplications/additions when $d_1 \ge d_0$~\cite[Theorem~9.6]{Gathen03}, where $\mathsf{M}(h)$ stands for the numbers of multiplications/additions required to multiply two polynomials of degree less than $h$.

\subsection{Fast Extended Euclidean Algorithm}
Let $r_0$ and $r_1$ be two monic polynomials with $\deg r_0 > \deg r_1$ and we assume $s_0 = t_1 = 1, s_1 = t_0 = 0$.
Step~$i$ ($i= 1, 2, \cdots, l$) of the EEA computes $\rho_{i+1} r_{i+1} = r_{i-1} - q_i r_i$, $\rho_{i+1} s_{i+1} = s_{i-1} - q_i s_i$, and $\rho_{i+1} t_{i+1} = t_{i-1} - q_i t_i$ so that the sequence $r_{i}$ are monic polynomials with strictly decreasing degrees.
If the GCD of $r_0$ and $r_1$ is desired, the EEA terminates when $r_{l+1}=0$.
For $1 \le i \le l$, $R_i \triangleq Q_i \cdots Q_1 R_0$, where $Q_i = \bigl[\begin{smallmatrix} 0 & 1\\ \frac{1}{\rho_{i+1}} & -\frac{q_i}{\rho_{i+1}}\end{smallmatrix}\bigr]$ and
$R_0 = \bigl[\begin{smallmatrix} 1 &0\\ 0 & 1\end{smallmatrix}\bigr]$.
	Then it can be easily verified that $R_i = \bigl[\begin{smallmatrix} s_i & t_i\\ s_{i+1} & t_{i+1}\end{smallmatrix}\bigr]$ for $0 \le i \le l$.
In RS decoding, the EEA stops when the degree of $r_i$ falls below a certain threshold for the first time, and we refer to this as partial GCD.

The FEEA in~\cite{Gathen03,Khodadad06} costs no more than $\bigl(22\mathsf{M}(h)+O(h)\bigr)\log{h}$ multiplications/additions when $n_0 \le 2h$~\cite{Khodadad05}.

\subsection{Syndrome-based and Syndromeless Decoding}
Over a finite field $\mathrm{GF}(q)$, suppose $a_0,a_1,\dots,a_{n-1}$ are $n$ ($n\leq q$) distinct elements and $g_0(x)\triangleq\prod_{i=0}^{n-1}(x-a_i)$.
Let us consider an RS code over $\mathrm{GF}(q)$ with length $n$, dimension $k$, and minimum Hamming distance $d = n - k + 1$.
A message polynomial $m(x)$ of degree less than $k$ is encoded to a codeword $(c_0, c_1, \cdots, c_{n-1})$ with $c_i = m(a_i)$, and the received vector is given by $\boldsymbol{r}=(r_0, r_1, \cdots, r_{n-1})$.

The syndrome-based hard decision decoding consists of the following steps: syndrome computation, key equation solver, the Chien search, and Forney's formula.
Further details are omitted, and interested readers are referred to~\cite{Berlekamp68,Wicker94,Moon05}.
We also consider the following two syndromeless algorithms:
\begin{algorithm}~\cite{Shiozaki88,Shiozaki90},~\cite[Algorithm 1]{Gao03}\label{alg:gao}
	\renewcommand{\labelenumi}{\ref{alg:gao}.\theenumi}
	\begin{enumerate}
	   \item Interpolation: Construct a polynomial $g_1(x)$ with $\deg g_1(x) < n$ such that $g_1(a_i)=r_i$ for $i=0,1,\dots,n-1$.
		\item Partial GCD:
			Apply the EEA to $g_0(x)$ and $g_1(x)$, and find $g(x)$ and $v(x)$ that maximize $\deg g(x)$ while satisfying $v(x)g_1(x)\equiv g(x) \bmod g_0(x)$ and $\deg g(x) < \frac{n+k}{2}$.
		\item Message Recovery:
			If $v(x) \mid g(x)$, the message polynomial is recovered by $m(x) = \frac{g(x)}{v(x)}$, otherwise output ``decoding failure.''
	\end{enumerate}
\end{algorithm}
\begin{algorithm}{\cite[Algorithm 1a]{Gao03}}
	\renewcommand{\labelenumi}{\ref{alg:gao_mod}.\theenumi}
	\begin{enumerate}
		\item Interpolation: Construct a polynomial $g_1(x)$ with $\deg g_1(x) < n$ such that $g_1(a_i)=r_i$ for $i=0,1,\dots,n-1$.
		\item Partial GCD:
			Find $s_0(x)$ and $s_1(x)$ satisfying $g_0(x) = x^{n-d+1} s_0(x) + r_0(x)$ and $g_1(x)=x^{n-d+1}s_1(x) + r_1(x)$, where $\deg r_0(x) \leq n-d$ and $\deg r_1(x) \leq n-d$.
			Apply the EEA to $s_0(x)$ and $s_1(x)$, and stop when the remainder $g(x)$ has degree less than $\frac{d-1}{2}$.
			Thus, we have $v(x)s_1(x) + u(x)s_0(x) = g(x)$.
		\item Message Recovery:
			If $v(x)\nmid g_0(x)$, output ``decoding failure''; otherwise, first compute $q(x) \triangleq \frac{g_0(x)}{v(x)}$, and then obtain $m'(x) = g_1(x) + q(x)u(x)$.
			If $\deg m'(x) < k$, output $m'(x)$; otherwise output ``decoding failure.''
	\end{enumerate}
\label{alg:gao_mod}
\end{algorithm}
Compared with Algorithm~\ref{alg:gao}, the partial GCD step of Algorithm~\ref{alg:gao_mod} is simpler but its message recovery step is more complex~\cite{Gao03}.

\section{Direct Implementation of Syndromeless Decoding}\label{sec:time}
\subsection{Complexity Analysis}\label{sec:time_ana}
We analyze the complexity of direct implementation of
Algorithms~\ref{alg:gao}~and~\ref{alg:gao_mod}.
For simplicity, we assume $n-k$ is even and hence $d-1=2t$.

First, $g_1(x)$ in Steps~\ref{alg:gao}.1~and~\ref{alg:gao_mod}.1 is given by $\mathrm{IDFT}(\boldsymbol{r})$.
Direct implementation of Steps~\ref{alg:gao}.1~and~\ref{alg:gao_mod}.1 follows Horner's rule, and requires  $n(n-1)$ multiplications and $n(n-1)$ additions~\cite{Komo02}.

Steps~\ref{alg:gao}.2~and~\ref{alg:gao_mod}.2 both use the EEA.
The Sugiyama tower (ST)~\cite{Sugiyama75,Berlekamp94} is well known as an efficient direct implementation of the EEA.
For Algorithm~\ref{alg:gao}, the ST is initialized by $g_1(x)$ and $g_0(x)$, whose degrees are at most $n$.
Since the number of iterations is $2t$, Step~\ref{alg:gao}.2 requires $4t(n+2)$ multiplications and $2t(n+1)$ additions.
For Algorithm~\ref{alg:gao_mod}, the ST is initialized by $s_0(x)$ and $s_1(x)$, whose degrees are at most $2t$ and the iteration number is at most $2t$.

Step~\ref{alg:gao}.3 requires one polynomial division, which can be implemented by using $k$ iterations of cross multiplications in the ST.
Since $v(x)$ is actually the error locator polynomial~\cite{Gao03}, $\deg v(x) \leq t$.
Hence, this requires $k(k+2t+2)$ multiplications and $k(t+2)$ additions.
However, the result of the polynomial division is scaled by a nonzero constant.
That is, cross multiplications lead to $\bar{m}(x)=am(x)$.
To remove the scaling factor $a$, we can first compute $\frac{1}{a} = \frac{\mathrm{lc}(g(x))}{\mathrm{lc}(\bar{m}(x)) \mathrm{lc}(v(x))}$, where $\mathrm{lc}(f)$ denotes the leading coefficient of a polynomial $f$, and then obtain $m(x)=\frac{1}{a}\bar{m}(x)$.
This process requires one inversion and $k+2$ multiplications.

Step~\ref{alg:gao_mod}.3 involves one polynomial division, one polynomial multiplication, and one polynomial addition, and their complexities depend on the degrees of $v(x)$ and $u(x)$, denoted as $d_v$ and $d_u$, respectively.
In the polynomial division, let the result of the ST be $\bar{q}(x)=aq(x)$.
The scaling factor is recovered by $\frac{1}{a} = \frac{1}{\mathrm{lc}(\bar{q}(x))\mathrm{lc}(v(x))}$.
Thus it requires one inversion, $(n-d_v + 1)(n+d_v+3)+n-d_v+2$ multiplications, and $(n - d_v + 1)(d_v +2)$ additions to obtain $q(x)$.
The polynomial multiplication needs $(n-d_v+1)(d_u+1)$ multiplications and $(n-d_v+1)(d_u+1)-(n-d_v + d_u + 1)$ additions, and the polynomial addition needs $n$ additions since $g_1(x)$ has degree at most $n-1$.
The total complexity of Step~\ref{alg:gao_mod}.3 includes $(n-d_v+1)(n+d_v+d_u+5)+1$ multiplications, $(n-d_v+1)(d_v+d_u+2)+n-d_u$ additions, and one inversion.
Consider the worst case for multiplicative complexity, where $d_v$ should be as small as possible.
But $d_v > d_u$, so the highest multiplicative complexity is $(n-d_u)(n+2d_u+6)+1$, which maximizes when $d_u = \frac{n-6}{4}$.
And we know $d_u < d_v \le t$.
Let $R$ denote the code rate.
So for RS codes with $R > \frac{1}{2}$, the maximum complexity is $n^2+nt-2t^2+5n-2t+5$ multiplications, $2nt-2t^2+2n+2$ additions, and one inversion.
For codes with $R \le \frac{1}{2}$, the maximum complexity is $\frac{9}{8}n^2+\frac{9}{2}n+\frac{11}{2}$ multiplications, $\frac{3}{8}n^2+\frac{3}{2}n+\frac{3}{2}$ additions, and one inversion.

Table~\ref{tab:gao_time} lists the complexity of direct implementation of Algorithms~\ref{alg:gao}~and~\ref{alg:gao_mod}, in terms of operations in $\mathrm{GF}(2^m)$.
\begin{table*}[htbp]
	\centering
	\caption{Direct Implementation Complexities of Syndromeless Decoding Algorithms}
	\begin{tabular}{|c|c|c|c|c|}
		\hline
		\multicolumn{2}{|c|}{} & Multiplications & Additions & Inversions\\
		\hline
		\multicolumn{2}{|c|}{Interpolation} & $n(n-1)$ & $n(n-1)$ & 0\\
		\hline
		\multirow{2}{*}{Partial GCD} & Algorithm~\ref{alg:gao} & $4t(n+2)$ & $2t(n+1)$ & 0\\
		\cline{2-5}
		& Algorithm~\ref{alg:gao_mod} & $4t(2t+2)$ & $2t(2t+1)$ & 0\\
		\hline
		\multirow{2}{*}{Message Recovery} & Algorithm~\ref{alg:gao} & $(k+2)(k+1)+2kt$ & $k(t+2)$ & 1\\
		\cline{2-5}
		& Algorithm~\ref{alg:gao_mod} & $n^2+nt-2t^2+5n-2t+5$ & $2nt-2t^2+2n+2$ & 1\\
		\hline
		\multirow{2}{*}{Total} & Algorithm~\ref{alg:gao} & $2n^2+2nt+2n+2t+2$ & $n^2+3nt-2t^2+n-2t$ & 1\\
		\cline{2-5}
		& Algorithm~\ref{alg:gao_mod} & $2n^2+nt+6t^2+4n+6t+5$ & $n^2+2nt+2t^2+n+2t+2$ & 1\\
		\hline
	\end{tabular}
	\label{tab:gao_time}
\end{table*}
The complexity of syndrome-based decoding is given in Table~\ref{tab:syn_tim}.
The numbers for syndrome computation, the Chien search, and Forney's formula are from~\cite{Mandelbaum71}.
We assume the EEA is used for the key equation solver since it was shown to be equivalent to the BMA~\cite{Heydtmann00}.
The ST is used to implement the EEA.
Note that the overall complexity of syndrome-based decoding can be reduced by sharing computations between the Chien search and Forney's formula.
However, this is not taken into account in Table~\ref{tab:syn_tim}.
\begin{table*}[htbp]
	\centering
	\caption{Direct Implementation Complexity of Syndrome-Based Decoding}
	\begin{tabular}{|c|c|c|c|}
		\hline
		& Multiplications & Additions & Inversions\\
		\hline
		Syndrome Computation & $2t(n-1)$ & $2t(n-1)$ & 0\\
		\hline
		Key Equation Solver & $4t(2t+2)$ & $2t(2t+1)$ & 0\\
		\hline
		Chien Search & $n(t-1)$ & $nt$ & 0\\
		\hline
		Forney's Formula & $2t^2$ & $t(2t-1)$ & $t$\\
		\hline
		Total & $3nt+10t^2-n+6t$ & $3nt+6t^2-t$ & $t$\\
		\hline
	\end{tabular}
	\label{tab:syn_tim}
\end{table*}

\subsection{Complexity Comparison}\label{sec:dir_comp}

For any application with fixed parameters $n$ and $k$, the comparison between the algorithms is straightforward using the complexities in Tables~\ref{tab:gao_time}~and~\ref{tab:syn_tim}.
Below we try to determine which algorithm is more suitable for a given code rate.
The comparison between different algorithms is complicated by three different types of field operations.
However, the complexity is dominated by the number of multiplications: in hardware implementation, both multiplication and inversion over $\GF(2^m)$ requires an area-time complexity of $O(m^2)$~\cite{Yan03}, whereas an addition requires an area-time complexity of $O(m)$; the complexity due to inversions is negligible since the required number of inversions is much smaller than those of multiplications; the numbers of multiplications and additions are both $O(n^2)$.
Thus, we focus on the number of multiplications for simplicity.

Since $t=\frac{1-R}{2}n$ and $k=Rn$, the multiplicative complexities of Algorithms~\ref{alg:gao}~and~\ref{alg:gao_mod} are $(3-R)n^2 + (3-R)n + 2$ and $\frac{1}{2}(3R^2-7R+8)n^2 + (7-3R)n + 5$, respectively, while the complexity of syndrome-based decoding is $\frac{5R^2-13R+8}{2}n^2 + (2-3R)n$.
It is easy to verify that in all these complexities, the quadratic and linear coefficients are of the same order of magnitude; hence, we consider only the quadratic terms.
Considering only the quadratic terms, Algorithm~\ref{alg:gao} is less efficient than syndrome-based decoding when $R > \frac{1}{5}$.
If the Chien search and Forney's formula share computations, this threshold will be even lower.
Comparing the highest terms, Algorithm~\ref{alg:gao_mod} is less efficient than
the syndrome-based algorithm regardless of $R$.
It is easy to verify that the most significant term of the difference between Algorithms~\ref{alg:gao}~and~\ref{alg:gao_mod} is $\frac{(1-R)(3R-2)}{2}n^2$.
So when implemented directly, Algorithm~\ref{alg:gao} is less efficient than Algorithm~\ref{alg:gao_mod} when $R > \frac{2}{3}$.
Thus, Algorithm~\ref{alg:gao} is more suitable for codes with very low rate, while syndrome-based decoding is the most efficient for high rate codes.

\subsection{Hardware Costs, Latency, and Throughput}\label{sec:dir_hw}
We have compared the computational complexities of syndromeless decoding algorithms with those of syndrome-based algorithms.
Now we compare these two types of decoding algorithms from a hardware perspective: we will compare the hardware costs, latency, and throughput of decoder architectures based on direct implementations of these algorithms.
Since our goal is to compare syndrome-based algorithms with syndromeless algorithms, we select our architectures so that the comparison is on a level field.
Thus, among various decoder architectures available for syndrome-based decoders in the literature, we consider the hypersystolic architecture in~\cite{Berlekamp94}.
Not only is it an efficient architecture for syndrome-based decoders, some of its functional units can be easily adapted to implement syndromeless decoders.
Thus, decoder architectures for both types of decoding algorithms have the same structure with some functional units the same; this allow us to focus on the difference between the two types of algorithms.
For the same reason, we do not try to optimize the hardware costs, latency, or throughput using circuit level techniques since such techniques will benefit the architectures for both types of decoding algorithms in a similar fashion and hence does not affect the comparison.

The hypersystolic architecture~\cite{Berlekamp94} contains three functional units: the power sums tower (PST) computing the syndromes, the ST solving the key equation, and the correction tower (CT) performing the Chien search and Forney's formula.
The PST consists of $2t$ systolic cells, each of which comprises of one multiplier, one adder, five registers, and one multiplexer.
The ST has $\delta+1$ ($\delta$ is the maximal degree of the input polynomials) systolic cells, each of which contains one multiplier, one adder, five registers, and seven multiplexers.
The latency of the ST is $6\gamma$ clock cycles~\cite{Berlekamp94}, where $\gamma$ is the number of iterations.
For the syndrome-based decoder architecture, $\delta$ and $\gamma$ are both $2t$.
The CT consists of $3t+1$ evaluation cells, two delay cells, along with two joiner cells, which also perform inversions.
Each evaluation cell needs one multiplier, one adder, four registers, and one multiplexer.
Each delay cell needs one register.
The two joiner cells altogether need two multipliers, one inverter, and four registers.
Table~\ref{tab:hw_syn} summarizes the hardware costs of the decoder architecture for syndrome-based decoders described above.
For each functional unit, we also list the latency (in clock cycles), as well as the number of clock cycles it needs to process one received word, which is proportional to the inverse of the throughput.
In theory, the computational complexities of steps of RS decoding depend on the received word, and the total complexity is obtained by first computing the \emph{sum} of complexities for all the steps and then considering the worst case scenario (cf.  Section~\ref{sec:time_ana}).
In contrast, the hardware costs, latency, and throughput of \emph{every} functional unit are dominated by the worst case scenario; the numbers in Table~\ref{tab:hw_syn} all correspond to the worst case scenario.
The critical path delay (CPD) is the same, $T_{mult} + T_{add} + T_{mux}$, for the PST, ST, and CT. In addition to the registers required by the PST, ST, and CT, the total number of registers in Table~\ref{tab:hw_syn} also account for the registers needed by the delay line called Main Street~\cite{Berlekamp94}.
\begin{table*}[htbp]
	\centering
	\caption{Decoder Architecture Based on Syndrome-Based Decoding (CPD is $T_{mult} + T_{add} + T_{mux}$)}
	\label{tab:hw_syn}
	\begin{tabular}{|c|c|c|c|c|c||c||c|}
		\hline
		& Multipliers & Adders & Inverters & Registers & Muxes & Latency & $\text{Throughput}^{-1}$\\
		\hline
		Syndrome Computation & $2t$ & $2t$ & 0 & $10t$ & $2t$ & $n+6t$ & $6t$\\
		\hline
		Key Equation Solver & $2t+1$ & $2t+1$ & 0 & $10t+5$ & $14t+7$ & $12t$ & $12t$\\
		\hline
		Correction & $3t+3$ & $3t+1$ & 1 & $12t+10$ & $3t+1$ & $3t$ & $3t$\\
		\hline
		Total & $7t+4$ & $7t+2$ & 1 & $n+53t+15$ & $19t+8$ & $n+21t$ & $12t$\\
		\hline
	\end{tabular}
\end{table*}

Both the PST and the ST can be adapted to implement decoder architectures for syndromeless decoding algorithms.
Similar to syndrome computation, interpolation in syndromeless decoders can be implemented by Horner's rule, and thus the PST can be easily adapted to implement this step.
For the architectures based on syndromeless decoding, the PST contains $n$ cells, and the hardware costs of each cell remain the same.
The partial GCD is implemented by the ST.
The ST can implement the polynomial division in message recovery as well.
In Step~\ref{alg:gao}.3, the maximum polynomial degree of the polynomial division is $k+t$ and the iteration number is at most $k$.
As mentioned in Section~\ref{sec:time_ana}, the degree of $q(x)$ in Step~\ref{alg:gao_mod}.3 ranges from $1$ to $t$.
In the polynomial division $\frac{g_0(x)}{v(x)}$, the maximum polynomial degree is $n$ and the iteration number is at most $n-1$.
Given the maximum polynomial degree and iteration number, the hardware costs and latency for the ST can be determined as for the syndrome-based architecture.

The other operations of syndromeless decoders do not have corresponding functional units available in the hypersystolic architecture, and we choose to implement them in a straightforward way.
In the polynomial multiplication $q(x)u(x)$, $u(x)$ has degree at most $t-1$ and the product has degree at most $n-1$.
Thus it can be done by $n$ multiply-and-accumulate circuits, $n$ registers in $t$ cycles (see, e.g.,~\cite{Park05}).
The polynomial addition in Step~\ref{alg:gao_mod}.3 can be done in one clock cycle with $n$ adders and $n$ registers.
To remove the scaling factor, Step~\ref{alg:gao}.3 is implemented in four cycles with at most one inverter, $k+2$ multipliers, and $k+3$ registers; Step~\ref{alg:gao_mod}.3 is implemented in three cycles with at most one inverter, $n+1$ multipliers, and $n+2$ registers.
We summarize the hardware costs, latency, and throughput of the decoder architectures based on Algorithms~\ref{alg:gao}~and~\ref{alg:gao_mod} in Table~\ref{tab:hw_gao}.
\begin{table*}[htbp]
	\centering
	\caption{Decoder Architectures Based on Syndromeless Decoding (CPD is $T_{mult} + T_{add} + T_{mux}$)}
	\label{tab:hw_gao}
	\small\addtolength{\tabcolsep}{-3pt}
	\scalebox{0.8}{
	\begin{tabular}{|c|c|c|c|c|c|c||c||c|}
		\hline
		\multicolumn{2}{|c|}{} & Multipliers & Adders & Inverters & Registers & Muxes & Latency & $\text{Throughput}^{-1}$\\
		\hline
		\multicolumn{2}{|c|}{Interpolation} & $n$ & $n$ & 0 & $5n$ & $n$ & $4n$ & $3n$\\
		\hline
		\multirow{2}{*}{Partial GCD} & Alg.~\ref{alg:gao} & $n+1$ & $n+1$ & 0 & $5n+5$ & $7n+7$ & $12t$ & $12t$\\
		\cline{2-9}
		 & Alg.~\ref{alg:gao_mod} & $2t+1$ & $2t+1$ & 0 & $10t+5$ & $14t+7$ & $12t$ & $12t$\\
		\hline
		Message & Alg.~\ref{alg:gao} & $2k+t+3$ & $k+t+1$ & 1 & $6k+5t+8$ & $7k+7t+7$ & $6k+4$ & $6k$\\
		\cline{2-9}
		 Recovery & Alg.~\ref{alg:gao_mod} & $3n+2$ & $3n+1$ & 1 & $7n+7$ & $7n+7$ & $6n+t-2$ & $6n$\\
		\hline
		\multirow{2}{*}{Total} & Alg.~\ref{alg:gao} & $2n+2k+t+4$ & $2n+k+t+2$ & 1 & $10n+6k+5t+13$ & $8n+7k+7t+14$ & $4n+6k+12t+4$ & $6k$\\
		\cline{2-9}
		 & Alg.~\ref{alg:gao_mod} & $4n+2t+3$ & $4n+2t+2$ & 1 & $12n+10t+12$ & $8n+14t+14$ & $10n+13t-2$ & $6n$\\
		\hline
	\end{tabular}
	}
\end{table*}

Now we compare the hardware costs of the three decoder architectures based on Tables~\ref{tab:hw_syn}~and~\ref{tab:hw_gao}.
The hardware costs are measured by the numbers of various basic circuit elements.
All three decoder architectures need only one inverter.
The syndrome-based decoder architecture requires fewer multiplexers than the decoder architecture based on Algorithm~\ref{alg:gao}, regardless of the rate, and fewer multipliers, adders, and registers when $R > \frac{1}{2}$.
The syndrome-based decoder architecture requires fewer registers than the decoder architecture based on Algorithm~\ref{alg:gao_mod} when $R > \frac{21}{43}$, and fewer multipliers, adders, and multiplexers regardless of the rate. 
Thus for high rate codes, the syndrome-based decoder has lower hardware costs than syndromeless decoders.
The decoder architecture based on Algorithm~\ref{alg:gao} requires fewer multipliers and adders than that based on Algorithm~\ref{alg:gao_mod}, regardless of the rate, but more registers and multiplexers when $R > \frac{9}{17}$.

In these algorithms, each step starts with the results of the previous step.
Due to this data dependency, their corresponding functional units have to operate in a pipelined fashion.
Thus the decoding latency is simply the sum of the latency of all the functional units.
The decoder architecture based on Algorithm~\ref{alg:gao_mod} has the longest latency, regardless of the rate.
The syndrome-based decoder architecture has shorter latency than the decoder architecture based on Algorithm~\ref{alg:gao} when $R>\frac{1}{7}$.

All three decoders have the same CPD, so the throughput is determined by the number of clock cycles.
Since the functional units in each decoder architecture are pipelined, the throughput of each decoder architecture is determined by the functional unit that requires the largest number of cycles.
Regardless of the rate, the decoder based on Algorithm~\ref{alg:gao_mod} has the lowest throughput.
When $R > \frac{1}{2}$, the syndrome-based decoder architecture has higher throughput than the decoder architecture based on Algorithm~\ref{alg:gao}.
When the rate is lower, they have the same throughput.

Hence for high rate RS codes, the syndrome-based decoder architecture requires less hardware and achieves higher throughput and shorter latency than those based on syndromeless decoding algorithms.

\section{Fast Implementation of Syndromeless Decoding}\label{sec:tran}
In this section, we implement the three steps of Algorithms~\ref{alg:gao}~and~\ref{alg:gao_mod}---interpolation, partial GCD, and message recovery---by fast algorithms described in Section~\ref{sec:rev} and evaluate their complexities.
Since both the polynomial division by Newton iteration and the FEEA depend on efficient polynomial multiplication, the decoding complexity relies on the complexity of polynomial multiplication.
Thus, in addition to field multiplications and additions, the complexities in this section are also expressed in terms of polynomial multiplications.

\subsection{Polynomial Multiplication}\label{sec:poly_mul}
We first derive a tighter bound on the complexity of the fast polynomial multiplication based on Cantor's approach.

Let the degree of the product of two polynomials be less than $n$.
The polynomial multiplication can be done by two FFTs and one inverse FFT if length-$n$ FFT is available over $\GF(2^m)$, which requires $n \mid 2^m-1$.
If $n \nmid 2^m-1$, one option is to pad the polynomials to length $n'$ ($n'> n$) with $n' \mid 2^m-1$.
Compared with fast polynomial multiplication based on multiplicative FFT, Cantor's approach uses additive FFT and does not require $n \mid 2^m-1$, so it is more efficient than FFT multiplication with padding for most degrees.
For $n=2^m-1$, their complexities are similar.
Although asymptotically worse than Sch\"onhage's algorithm~\cite{Gathen03}, which has $O(n\log n \log\log n)$ complexity, Cantor's approach has small implicit constants and hence it is more suitable for practical implementation of RS codes~\cite{Gao03,Gathen96a}.
Gao claimed an improvement on Cantor's approach in~\cite{Gao03}, but we do not pursue this due to lack of details.

A tighter bound on the complexity of Cantor's approach is given in Theorem~\ref{trm:poly_mul}.
Here we make the same assumption as in~\cite{Gathen96a} that the auxiliary polynomials $s_i$ and the values $s_i(\beta_j)$ are precomputed.
The complexity of pre-computation was given in~\cite{Gathen96a}.
\begin{theorem}\label{trm:poly_mul}
	By Cantor's approach, two polynomials $a,b \in \GF(2^m)[x]$ whose product has degree less than $h$ ($1 \le h \le 2^m$) can be multiplied using less than $\frac{3}{2}h\log^2h+\frac{7}{2}h\log h-2h+\log h+2$ multiplications, $\frac{3}{2}h\log^2h+\frac{21}{2}h\log h-13h+\log h+15$ additions, and $2h$ inversions over $\GF(2^m)$.
\end{theorem}
\begin{proof}
	There exists $0 \le p \le m$ satisfying $2^{p-1} < h \le 2^p$.
	Since both the MPE and MPI algorithms are recursive, we denote the numbers of additions of the MPE and MPI algorithms for input $i$ ($0\leq i \leq p)$ as $S_E(i)$ and $S_I(i)$, respectively.
	Clearly $S_E(0)=S_I(0)=0$.
	Following the approach in~\cite{Gathen96a}, it can be shown that for $1 \leq i \leq p$,
	\begin{align}
		S_E(i) & \le  i(i+3)2^{i-2} + (p-3)(2^i-1) + i, \label{eqn:se}\\
		S_I(i) & \le  i(i+5)2^{i-2} + (p-3)(2^i-1) + i. \label{eqn:si}
	\end{align}

Let $M_E(h)$ and $A_E(h)$ denote the numbers of multiplications and additions, respectively, that the MPE algorithm requires for polynomials of degree less than $h$.
When $i=p$ in the MPE algorithm, $f(x)$ has degree less than $h \le 2^p$, while $s_{p-1}$ is of degree $2^{p-1}$ and has at most $p$ non-zero coefficients.
Thus $g(x)$ has degree less than $h -2^{p-1}$.
Therefore the numbers of multiplications and additions for the polynomial division in~\cite[Step~2 of Algorithm~3.1]{Gathen96a} are both $p(h-2^{p-1})$, while $r_1(x) = r_0(x) + s_{i-1}(\beta_i)g(x)$ needs at most $h-2^{p-1}$ multiplications and the same number of additions.
Substituting the bound on $M_E(2^{p-1})$ in~\cite{Gathen96a}, we obtain $M_E(h) \leq 2M_E(2^{p-1}) + p(h-2^{p-1}) +h-2^{p-1}$, and thus $M_E(h)$ is at most $\frac{1}{4}p^22^p-\frac{1}{4}p2^p-2^p+(p+1)h$. Similarly, substituting the bound on $S_E(p-1)$ in Eq.~\eqref{eqn:se}, we obtain $A_E(h) \leq 2S_E(p-1) + p(h-2^{p-1}) +h-2^{p-1}$, and hence $A_E(h)$ is at most $\frac{1}{4}p^22^p+\frac{3}{4}p2^p-4 \cdot 2^p+(p+1)h+4$.

Let $M_I(h)$ and $A_I(h)$ denote the numbers of multiplications and additions, respectively, the MPI algorithm requires when the interpolated polynomial has degree less than $h$.
When $i=p$ in the MPI algorithm, $f(x)$ has degree less than $h \le 2^p$.
It implies that $r_0(x)+r_1(x)$ has degree less than $h-2^{p-1}$.
Thus it requires at most $h-2^{p-1}$ additions to obtain $r_0(x)+r_1(x)$ and $h-2^{p-1}$ multiplications for $s_{i-1}(\beta_i)^{-1}\bigl(r_0(x)+r_1(x)\bigr)$.
The numbers of multiplications and additions for the polynomial multiplication in~\cite[Step~3 of Algorithm~3.2]{Gathen96a} to obtain $f(x)$ are both $p(h-2^{p-1})$.
Adding $r_0(x)$ also needs $2^{p-1}$ additions.
Substituting the bound on $M_I(2^{p-1})$ in~\cite{Gathen96a}, we have $M_I(h) \le 2M_I(2^{p-1}) + p(h-2^{p-1})+h-2^{p-1}$, and hence $M_I(h)$ is at most $\frac{1}{4}p^22^p - \frac{1}{4}p2^p -2^p+(p+1)h$. Similarly, substituting the bound on $S_I(p-1)$ in Eq.~\eqref{eqn:si}, we have $A_I(h) \le  2S_I(p-1) + p(h-2^{p-1})+h+1$, and hence $A_E(h)$ is at most $\frac{1}{4}p^22^p+\frac{5}{4}p2^p-4\cdot2^p+(p+1)h+5$.
The interpolation step also needs $2^p$ inversions.

Let $\mathsf{M}(h_1,h_2)$ be the complexity of multiplication of two polynomials of degrees less than $h_1$ and $h_2$.
Using Cantor's approach, $\mathsf{M}(h_1,h_2)$ includes $M_E(h_1)+M_E(h_2)+M_I(h)+2^p$ multiplications, $A_E(h_1)+A_E(h_2)+A_I(h)$ additions, and $2^p$ inversions, when $h=h_1+h_2-1$.
Finally, we replace $2^p$ by $2h$ as in~\cite{Gathen96a}.
\end{proof}

Compared with the results in~\cite{Gathen96a}, our results have the same highest degree term but smaller terms for lower degrees.

By Theorem~\ref{trm:poly_mul}, we can easily compute $\mathsf{M}(h_1) \triangleq \mathsf{M}(h_1,h_1)$.
A by-product of the above proof is the bounds for the MPE and MPI algorithms.
We also observe some properties for the complexity of fast polynomial multiplication that hold for not only Cantor's approach but also other approaches.
These properties will be used in our complexity analysis next.
Since all fast polynomial multiplication algorithms have higher-than-linear complexities, $2\mathsf{M}(h) \le \mathsf{M}(2h)$.
Also note that $\mathsf{M}(h+1)$ is no more than $\mathsf{M}(h)$ plus $2h$ multiplications and $2h$ additions~\cite[Exercise 8.34]{Gathen03}.
Since the complexity bound is determined only by the degree of the product polynomial, we assume $\mathsf{M}(h_1,h_2)\le\mathsf{M}(\lceil\frac{h_1+h_2}{2}\rceil)$.
We note that the complexities of Sch\"onhage's algorithm as well as Sch\"onhage and Strassen's algorithm, both based on multiplicative FFT, are also determined by the degree of the product polynomial~\cite{Gathen03}.

\subsection{Polynomial Division}\label{sec:poly_div}
Similar to~\cite[Exercise 9.6]{Gathen03}, in characteristic-2 fields, the complexity of Newton iteration is at most $$\sum_{0\le j \le r-1}\bigl(\mathsf{M}(\lceil (d_0+1) 2^{-j} \rceil) +\mathsf{M}(\lceil (d_0+1) 2^{-j-1} \rceil)\bigr),$$ where $r = \lceil \log(d_0 + 1)\rceil$.
Since $\lceil (d_0 + 1)2^{-j} \rceil \le \lfloor (d_0 +1)2^{-j} \rfloor +1$ and $\mathsf{M}(h+1)$ is no more than $\mathsf{M}(h)$, plus $2h$ multiplications and $2h$ additions~\cite[Exercise 8.34]{Gathen03}, it requires at most $\sum_{1\le j \le r}\bigl(\mathsf{M}(\lfloor (d_0+1) 2^{-j} \rfloor) + \mathsf{M}(\lfloor (d_0+1) 2^{-j-1} \rfloor)\bigr)$, plus $\sum_{0 \le j \le r-1}(2\lfloor (d_0+1) 2^{-j} \rfloor + 2\lfloor (d_0+1) 2^{-j-1} \rfloor)$ multiplications and the same
number of additions.
Since $2\mathsf{M}(h) \le \mathsf{M}(2h)$, Newton iteration costs at most $\sum_{0\le j \le r-1}\bigl(\frac{3}{2}\mathsf{M}(\lfloor (d_0+1) 2^{-j} \rfloor)\bigr) \le 3\mathsf{M}(d_0+1)$, $6(d_0+1)$ multiplications, and $6(d_0+1)$ additions.
The second step to compute the quotient needs $\mathsf{M}(d_0+1)$ and the last step to compute the remainder needs $\mathsf{M}(d_1+1,d_0+1)$ and $d_1+1$ additions.
By $\mathsf{M}(d_1+1,d_0+1) \le \mathsf{M}(\lceil\frac{d_0+d_1}{2}\rceil+1)$, the total cost is at most $4\mathsf{M}(d_0)+\mathsf{M}(\lceil\frac{d_0+d_1}{2}\rceil)$, $15d_0+d_1+7$ multiplications, and $11d_0+2d_1+8$ additions.
Note that this bound does not require $d_1\ge d_0$ as in~\cite{Gathen03}.

\subsection{Partial GCD}\label{sec:feea}
The partial GCD step can be implemented in three approaches: the ST, the classical EEA with fast polynomial multiplication and Newton iteration, and the FEEA with fast polynomial multiplication and Newton iteration.
The ST is essentially the classical EEA.
The complexity of the classical EEA is asymptotically worse than that of the FEEA.
Since the FEEA is more suitable for long codes, we will use the FEEA in our complexity analysis of fast implementations.

In order to derive a tighter bound on the complexity of the FEEA, we first present a modified FEEA in Algorithm~\ref{alg:feea}.
Let $\eta(h) \triangleq \max\{j:\sum_{i=1}^j \deg q_i \le h\}$, which is the number of steps of the EEA satisfying $\deg r_0 - \deg r_{\eta(h)} \le h < \deg r_0 - \deg r_{\eta(h)+1}$.
For $f(x) = f_n x^n + \cdots + f_1 x + f_0$ with $f_n \ne 0$, the truncated polynomial $f(x) \upharpoonright h \triangleq f_n x^h + \cdots + f_{n-h+1} x + f_{n-h}$ where $f_i=0$ for $i<0$.
Note that $f(x) \upharpoonright h = 0$ if $h < 0$.
\begin{algorithm}
	Modified Fast Extended Euclidean Algorithm
	\label{alg:feea}

	\KwIn{two \emph{monic} polynomials $r_0$ and $r_1$, with $\deg r_0 = n_0 > n_1 = \deg r_1$, as well as integer $h$ $(0 < h \le n_0)$}
	\KwOut{$l = \eta(h), \rho_{l+1}, R_l, r_l$, and $\tilde{r}_{l+1}$}
	\renewcommand{\labelenumi}{\ref{alg:feea}.\theenumi}
	\begin{enumerate}
		\item If $r_1 = 0$ or $h < n_0 - n_1$, then return $0, 1, \bigl[\begin{smallmatrix} 1 & 0\\ 0 & 1 \end{smallmatrix}\bigr], r_0$, and $r_1$.
		\item $h_1 = \lfloor \frac{h}{2} \rfloor, r_0^* = r_0 \upharpoonright 2h_1, r_1^* = r_1 \upharpoonright \bigl(2h_1 - (n_0 - n_1)\bigr)$.
		\item $(j - 1, \rho_j^*, R_{j-1}^*, r^*_{j-1}, \tilde{r}^*_j) = \mathrm{FEEA}(r_0^*, r_1^*, h_1)$.
		\item $\bigl[\begin{smallmatrix}r_{j-1}\\\tilde{r}_j \end{smallmatrix}\bigr] = R^*_{j-1}\bigl[\begin{smallmatrix}r_0-r^*_0x^{n_0-2h_1}\\ r_1-r_1^*x^{n_0-2h_1}\end{smallmatrix}\bigr] + \bigl[\begin{smallmatrix}r^*_{j-1}x^{n_0-2h_1}\\ \tilde{r}^*_j x^{n_0-2h_1}\end{smallmatrix}\bigr]$,
				$R_{j-1} = \bigl[\begin{smallmatrix} 1 & 0\\ 0 & \frac{1}{\mathrm{lc}(\tilde{r}_j)}\end{smallmatrix}\bigr]R_{j-1}^*$,
				$\rho_j = \rho_j^* \mathrm{lc}(\tilde{r}_j)$, $r_j = \frac{\tilde{r}_j}{\mathrm{lc}(\tilde{r}_j)}$, $n_j = \deg r_j$.
		\item If $r_j = 0$ or $h < n_0 - n_j$, then
			return $j - 1, \rho_j, R_{j-1}, r_{j-1}$, and $\tilde{r}_j$.
		\item Perform polynomial division with remainder as $r_{j-1} = q_j r_j + \tilde{r}_{j+1}$, $\rho_{j+1} = \mathrm{lc}(\tilde{r}_{j+1}),
			r_{j+1} = \frac{\tilde{r}_{j+1}}{\rho_{j+1}}, n_{j+1} = \deg r_{j+1},
			R_j = \bigl[\begin{smallmatrix} 0 & 1\\ \frac{1}{\rho_{j+1}} & -\frac{q_j}{\rho_{j+1}} \end{smallmatrix}\bigr] R_{j-1}$.
		\item $h_2 = h - (n_0 - n_j),
			r_j^* = r_j \upharpoonright 2h_2, r_{j+1}^* = r_{j+1} \upharpoonright \bigl(2h_2 - (n_j - n_{j+1})\bigr)$.
		\item $(l-j, \rho_{l+1}^*, S^*, r^*_{l-j},\tilde{r}^*_{l-j+1}) = \mathrm{FEEA}(r_j^*, r_{j+1}^*, h_2)$.
		\item $\bigl[\begin{smallmatrix} r_l\\ \tilde{r}_{l+1}\end{smallmatrix}\bigr] =
			S^*\bigl[\begin{smallmatrix}r_j-r_j^* x^{n_j-2h_2}\\ r_{j+1}-r_{j+1}^*x^{n_j-2h_2}\end{smallmatrix}\bigr] + \bigl[\begin{smallmatrix}r_{l-j}^*x^{n_j-2h_2}\\ \tilde{r}_{l-j+1}^*x^{n_j-2h_2}\end{smallmatrix}\bigr], S = \bigl[\begin{smallmatrix} 1 & 0\\ 0 & \frac{1}{\mathrm{lc}(\tilde{r}_{l+1})}\end{smallmatrix}\bigr] S^*,
					\rho_{l+1}=\rho_{l+1}^*\mathrm{lc}(\tilde{r}_{l+1})$.
				\item Return $l, \rho_{l+1}, SR_j, r_l, \tilde{r}_{l+1}$.
	\end{enumerate}
\end{algorithm}
It is easy to verify that Algorithm~\ref{alg:feea} is equivalent to the FEEA in~\cite{Gathen03,Khodadad06}.
The difference between Algorithm~\ref{alg:feea} and the FEEA in~\cite{Gathen03,Khodadad06} lies in Steps~\ref{alg:feea}.4,~\ref{alg:feea}.5,~\ref{alg:feea}.8,~and~\ref{alg:feea}.10: in Steps~\ref{alg:feea}.5~and~\ref{alg:feea}.10, two additional polynomials are returned, and they are used in the updates of Steps~\ref{alg:feea}.4~and~\ref{alg:feea}.8 to reduce complexity.
The modification in Step~\ref{alg:feea}.4 was suggested in~\cite{Khodadad05} and the modification in Step~\ref{alg:feea}.9 follows the same idea.

In~\cite{Gathen03,Khodadad05}, the complexity bounds of the FEEA are established assuming $n_0 \le 2h$.
Thus we first establish a bound of the FEEA for the case $n_0 \le 2h$ below in Theorem~\ref{trm:feea}, using the bounds we develop in Sections~\ref{sec:poly_mul}~and~\ref{sec:poly_div}.
The proof is similar to those in~\cite{Gathen03,Khodadad05} and hence omitted; interested readers should have no difficulty filling in the details.

\begin{theorem}\label{trm:feea}
	Let $T(n_0, h)$ denote the complexity of the FEEA.
	When $n_0 \le 2h$, $T(n_0, h)$ is at most $17\mathsf{M}(h)\log h$ plus $(48h+2)\log h$ multiplications, $(51h+2)\log h$ additions, and $3h$ inversions.
	Furthermore, if the degree sequence is normal, $T(2h,h)$ is at most $10\mathsf{M}(h)\log h$, $(\frac{55}{2}h+6)\log h$ multiplications, and $(\frac{69}{2}h+3)\log h$ additions.
\end{theorem}

Compared with the complexity bounds in~\cite{Gathen03,Khodadad05}, our bound not only is tighter, but also specifies all terms of the complexity and avoid the big $O$ notation.
The saving over~\cite{Khodadad05} is due to lower complexities of Steps~\ref{alg:feea}.6,~\ref{alg:feea}.9,~and~\ref{alg:feea}.10 as explained above.

The saving for the normal case over~\cite{Gathen03} is due to lower complexity of Step~\ref{alg:feea}.9.

Applying the FEEA to $g_0(x)$ and $g_1(x)$ to find $v(x)$ and $g(x)$ in Algorithm~\ref{alg:gao}, we have $n_0 = n$ and $h \le t$ since $\deg v(x) \le t$.
For RS codes, we always have $n > 2t$.
Thus, the condition $n_0 \le 2h$ for the complexity bound in~\cite{Gathen03,Khodadad05} is not valid.
It was pointed out in~\cite{Gathen03,Gao03} that $s_0(x)$ and $s_1(x)$ as defined in Algorithm~\ref{alg:gao_mod} can be used instead of $g_0(x)$ and $g_1(x)$, which is the difference between Algorithms~\ref{alg:gao}~and~\ref{alg:gao_mod}.
Although such a transform allows us to use the results in~\cite{Gathen03,Khodadad05}, it introduces extra cost for message recovery~\cite{Gao03}.
To compare the complexities of Algorithms~\ref{alg:gao}~and~\ref{alg:gao_mod}, we establish a more general bound in Theorem~\ref{trm:gfeea}.

\begin{theorem}
	The complexity of FEEA is no more than $34\mathsf{M}(\lfloor\frac{h}{2}\rfloor)\log\lfloor\frac{h}{2}\rfloor + \mathsf{M}(\lfloor \frac{n_0}{2} \rfloor) + 4\mathsf{M}(\lceil \frac{n_0}{2} - \frac{h}{4} \rceil) + 2\mathsf{M}(\lfloor\frac{n_0-h}{2}\rfloor) + 4\mathsf{M}(h) + 2\mathsf{M}(\lfloor\frac{3}{4}h\rfloor) + 4\mathsf{M}(\lfloor\frac{h}{2}\rfloor)$, $(48h+4)\log\lfloor\frac{h}{2}\rfloor + 9n_0 + 22h$ multiplications, $(51h + 4)\log\lfloor\frac{h}{2}\rfloor + 11n_0 + 17h + 2$ additions, and $3h$ inversions.
	\label{trm:gfeea}
\end{theorem}
The proof is also omitted for brevity. 
The main difference between this case and Theorem~\ref{trm:feea} lies in the top level call of the FEEA.
The total complexity is obtained by adding $2T(h, \lfloor \frac{h}{2} \rfloor)$ and the top-level cost.

It can be verified that, when $n_0\le 2h$, Theorem~\ref{trm:gfeea} presents a tighter bound than Theorem~\ref{trm:feea} since saving on the top level is accounted for.
Note that the complexity bounds in Theorems~\ref{trm:feea}~and~\ref{trm:gfeea} assume that the FEEA solves $s_{l+1}r_0+t_{l+1}r_1=\tilde{r}_{l+1}$ for both $t_{l+1}$ and $s_{l+1}$.
If $s_{l+1}$ is not necessary, the complexity bounds in Theorems~\ref{trm:feea}~and~\ref{trm:gfeea} are further reduced by $2\mathsf{M}(\lfloor\frac{h}{2}\rfloor)$, $3h+1$ multiplications, and $4h+1$ additions.

\subsection{Complexity Comparison}\label{sec:fast_comp}

Using the results in Sections~\ref{sec:poly_mul},~\ref{sec:poly_div},~and~\ref{sec:feea}, we first analyze and then compare the complexities of Algorithms~\ref{alg:gao}~and~\ref{alg:gao_mod} as well as syndrome-based decoding under fast implementations.

In Steps~\ref{alg:gao}.1~and~\ref{alg:gao_mod}.1, $g_1(x)$ can be obtained by an inverse FFT when $n | 2^m-1$ or by the MPI algorithm.
In the latter case, the complexity is given in Section~\ref{sec:poly_mul}.
By Theorem~\ref{trm:gfeea}, the complexity of Step~\ref{alg:gao}.2 is $T(n, t)$ minus the complexity to compute $s_{l+1}$.
The complexity of Step~\ref{alg:gao_mod}.2 is $T(2t,t)$.
The complexity of Step~\ref{alg:gao}.3 is given by the bound in Section~\ref{sec:poly_div}.
Similarly, the complexity of Step~\ref{alg:gao_mod}.3 is readily obtained by using the bounds of polynomial division and multiplication.

All the steps of syndrome-based decoding can be implemented using fast algorithms.
Both syndrome computation and the Chien search can be done by $n$-point evaluations.
Forney's formula can be done by two $t$-point evaluations plus $t$ inversions and $t$ multiplications.
To use the MPE algorithm, we choose to evaluate on all $n$ points.
By Theorem~\ref{trm:gfeea}, the complexity of the key equation solver is $T(2t,t)$ minus the complexity to compute $s_{l+1}$.

Note that to simplify the expressions, the complexities are expressed in terms of three kinds of operations: polynomial multiplications, field multiplications, and field additions.
Of course, with our bounds on the complexity of polynomial multiplication in Theorem~\ref{trm:poly_mul}, the complexities of the decoding algorithms can be expressed in terms of field multiplications and additions.

Given the code parameters, the comparison among these algorithms is quite straightforward with the above expressions.
As in Section~\ref{sec:dir_comp}, we attempt to compare the complexities using only $R$.
Such a comparison is of course not accurate, but it sheds light on the comparative complexity of these decoding algorithms without getting entangled in the details.
To this end, we make four assumptions.
First, we assume the complexity bounds on the decoding algorithms as approximate decoding complexities.
Second, we use the complexity bound in Theorem~\ref{trm:poly_mul} as approximate polynomial multiplication complexities.
Third, since the numbers of multiplications and additions are of the same degree, we only compare the numbers of multiplications.
Fourth, we focus on the difference of the second highest degree terms since the highest degree terms are the same for all three algorithms.
This is because the partial GCD steps of Algorithms~\ref{alg:gao}~and~\ref{alg:gao_mod}, as well as the key equation solver in syndrome-based decoding differ only in the top level of the recursion of FEEA.
Hence Algorithms~\ref{alg:gao}~and~\ref{alg:gao_mod} as well as the key equation solver in syndrome-based decoding have the same highest degree term.

We first compare the complexities of Algorithms~\ref{alg:gao}~and~\ref{alg:gao_mod}.
Using Theorem~\ref{trm:poly_mul}, the difference between the second highest degree terms is given by $\frac{3}{4}(25R-13)n\log^2n$, so Algorithm~\ref{alg:gao} is less efficient than Algorithm~\ref{alg:gao_mod} when $R > 0.52$.
Similarly, the complexity difference between syndrome-based decoding and Algorithm~\ref{alg:gao} is given by $\frac{3}{4}(1-31R)n\log^2n$.
Thus syndrome-based decoding is more efficient than Algorithm~\ref{alg:gao} when $R > 0.032$.
Comparing syndrome-based decoding and Algorithm~\ref{alg:gao_mod}, the complexity difference is roughly $-\frac{9}{2}(2+R)n\log^2n$.
Hence syndrome-based decoding is more efficient than Algorithm~\ref{alg:gao_mod} regardless of the rate.

We remark that the conclusion of the above comparison is similar to those obtained in Section~\ref{sec:dir_comp} except the thresholds are different.
Based on fast implementations, Algorithm~\ref{alg:gao} is more efficient than Algorithm~\ref{alg:gao_mod} for low rate codes, and the syndrome-based decoding is more efficient than Algorithms~\ref{alg:gao}~and~\ref{alg:gao_mod} in virtually all cases.

\section{Case Study and Discussions}\label{sec:case}

\subsection{Case Study}
We examine the complexities of Algorithms~\ref{alg:gao}~and~\ref{alg:gao_mod} as well as syndrome-based decoding for the $(255, 223)$ CCSDS RS code~\cite{CCSDS101} and a $(511,447)$ RS code which have roughly the same rate $R = 0.87$.
Again, both direct and fast implementations are investigated.
Due to the moderate lengths, in some cases direct implementation leads to lower complexity, and hence in such cases, the complexity of direct implementation is used for both.

Tables~\ref{tab:ccsds_gao}~and~\ref{tab:ccsds_bm} list the total decoding complexity of Algorithms~\ref{alg:gao}~and~\ref{alg:gao_mod} as well as syndrome-based decoding, respectively. 
In the fast implementations, cyclotomic FFT~\cite{Chen07} is used for interpolation, syndrome computation, and the Chien search.
The classical EEA with fast polynomial multiplication and division is used in fast implementations since it is more efficient than the FEEA for these lengths.
We assume normal degree sequence, which represents the worst case scenario~\cite{Gathen03}.
The message recovery steps use long division in fast implementation since it is more efficient than Newton iteration for these lengths.
We use Horner's rule for Forney's formula in both direct and fast implementations.

We note that for each decoding step, Tables~\ref{tab:ccsds_gao}~and~\ref{tab:ccsds_bm} not only provide the numbers of finite field multiplications, additions, and inversions, but also list the overall complexities to facilitate comparisons. 
The overall complexities are computed based on the assumptions that multiplication and inversion are of equal complexity, and that as in~\cite{Wang88}, one multiplication is equivalent to $2m$ additions.
The latter assumption is justified by both hardware and software implementation of finite field operations. 
In hardware implementation, a multiplier over $\mathrm{GF}(2^m)$ generated by trinomials requires $m^2-1$ XOR and $m^2$ AND gates~\cite{Sunar99}, while an adder requires $m$ XOR gates.
Assuming that XOR and AND gates have the same complexity, the complexity of a multiplier is $2m$ times that of an adder over $\mathrm{GF}(2^m)$.
In software implementation, the complexity can be measured by the number of word-level operations~\cite{Mahboob05}.
Using the shift and add method as in~\cite{Mahboob05}, a multiplication requires
$m-1$ shift and $m$ XOR word-level operations, respectively while an addition needs only one XOR word-level operation.
Henceforth in software implementations the complexity of a multiplication over $\mathrm{GF}(2^m)$ is also roughly $2m$ times as that of an addition. 
Thus the total complexity of each decoding step in Tables~\ref{tab:ccsds_gao}~and~\ref{tab:ccsds_bm} is obtained by $ N = 2m (N_{mult} + N_{inv}) + N_{add}$, which is in terms of field additions.

	\begin{table*}[htbp]
		\centering
		\caption{Complexity of Syndromeless Decoding}
		\small\addtolength{\tabcolsep}{-3pt}
		\scalebox{0.8}{
		\begin{tabular}{|c|c|c|c|c||c|c|c|c||c|c|c|c||c|c|c|c||c|}
			\hline
			\multicolumn{2}{|c|}{\multirow{3}{*}{$(n,k)$}} & \multicolumn{8}{c|}{Direct Implementation} & \multicolumn{8}{c|}{Fast Implementation}\\
			\cline{3-18}
			\multicolumn{2}{|c|}{} & \multicolumn{4}{c|}{Algorithm~\ref{alg:gao}} & \multicolumn{4}{c|}{Algorithm~\ref{alg:gao_mod}} & \multicolumn{4}{c|}{Algorithm~\ref{alg:gao}} & \multicolumn{4}{c|}{Algorithm~\ref{alg:gao_mod}}\\
			\cline{3-18}
			\multicolumn{2}{|c|}{} & Mult. & Add. & Inv. & Overall & Mult. & Add. & Inv. & Overall & Mult. & Add. & Inv. & Overall & Mult. & Add. & Inv. & Overall\\
			\hline
			\multirow{4}{*}{$(255,223)$} & Interpolation & 64770 & 64770 & 0 & 1101090 & 64770 & 64770 & 0 & 11101090 & 586 & 6900 & 0 & 16276 & 586 & 6900 & 0 & 16276\\
			\cline{2-18}
			& Partial GCD & 16448 & 8192 & 0 & 271360 & 2176 & 1056 & 0 & 35872 & 8224 & 8176 & 16 & 140016 & 1392 & 1328 & 16 & 23856\\
			\cline{2-18}
			& Msg Recovery & 57536 & 4014 & 1 & 924606 & 69841 & 8160 & 1 & 1125632 & 3791 & 3568 & 1 & 64240 & 8160 & 7665  & 1 & 138241\\
			\cline{2-18}
			& Total & 138754 & 76976 & 1 & 2297056 & 136787 & 73986 & 1 & 2262594 & 12601 & 18644 & 17 & 220532 & 10138 & 15893 & 17 & 178373\\
			\hline
			\multirow{4}{*}{$(511,447)$} & Interpolation & 260610 & 260610 & 0 & 4951590 & 260610 & 260610 & 0 & 4951590 & 1014 & 23424 & 0 & 41676 & 1014 & 23424 & 0 & 41676\\
			\cline{2-18}
			& Partial GCD & 65664 & 32768 & 0 & 1214720 & 8448 & 4160 & 0 & 156224 & 32832 & 32736 & 32 & 624288 & 5344 & 5216 & 32 & 101984\\
			\cline{2-18}
			& Msg Recovery & 229760 & 15198 & 1 & 4150896 & 277921 & 31680 & 1 & 5034276 & 14751 & 14304 & 1 & 279840 & 31680 & 30689 & 1 & 600947\\
			\cline{2-18}
			& Total & 556034 & 308576 & 1 & 10317206 & 546979 & 296450 & 1 & 10142090 & 48597 & 70464 & 33 & 945804 & 38038 & 59329 & 33 & 744607\\
			\hline
		\end{tabular}
		}
		\label{tab:ccsds_gao}
	\end{table*}

	\begin{table*}[htbp]
		\centering
		\caption{Complexity of Syndrome-Based Decoding}
		\begin{tabular}{|c|c|c|c|c||c|c|c|c||c|}
			\hline
			\multicolumn{2}{|c|}{\multirow{2}{*}{$(n,k)$}} & \multicolumn{4}{c|}{Direct Implementation} & \multicolumn{4}{c|}{Fast Implementation}\\
			\cline{3-10}
			\multicolumn{2}{|c|}{} & Mult. & Add. & Inv. & Overall & Mult. & Add. & Inv. & Overall\\
			\hline
			\multirow{5}{*}{(255, 223)} & Syndrome Computation & 8128 & 8128 & 0 & 138176 & 149 & 4012 & 0 & 6396\\
			\cline{2-10}
			& Key Equation Solver & 2176 & 1056 & 0 & 35872 & 1088 & 1040 & 16 & 18704\\
			\cline{2-10}
			& Chien Search & 3825 & 4080 & 0 & 65280 & 586 & 6900 & 0 & 16276\\
			\cline{2-10}
			& Forney's Formula & 512 & 496 & 16 & 8944 & 512 & 496 & 16 & 8944\\
			\cline{2-10}
			& Total & 14641 & 13760 & 16 & 248272 & 2335 & 12448 & 32 & 50320\\
			\hline
			\multirow{5}{*}{$(511,447)$} & Syndrome Computation & 32640 & 32640 & 0 & 620160 & 345 & 16952 & 0 & 23162\\
			\cline{2-10}
			& Key Equation Solver & 8448 & 4160 & 0 & 156224 & 4224 & 4128 & 32 & 80736\\
			\cline{2-10}
			& Chien Search & 15841 & 16352 & 0 & 301490 & 1014 & 23424 & 0 & 41676\\
			\cline{2-10}
			& Forney's Formula & 2048 & 2016 & 32 & 39456 & 2048 & 2016 & 32 & 39456\\
			\cline{2-10}
			& Total & 58977 & 55168 & 32 & 1117330 & 7631 & 46520 & 64 & 185030\\
			\hline
		\end{tabular}
		\label{tab:ccsds_bm}
	\end{table*}

Comparisons between direct and fast implementations for each algorithm show that fast implementations considerably reduce the complexities of both syndromeless and syndrome-based decoding, as shown in Tables~\ref{tab:ccsds_gao}~and~\ref{tab:ccsds_bm}.
The comparison between these tables show that for these two high-rate codes, both direct and fast implementations of syndromeless decoding are not as efficient as their counterparts of syndrome-based decoding.
This observation is consistent with our conclusions in Sections~\ref{sec:dir_comp}~and~\ref{sec:fast_comp}.

For these two codes, hardware costs and throughput of decoder architectures based on direct implementations of syndrome-based and syndromeless decoding can be easily obtained by substituting the parameters in Tables~\ref{tab:hw_syn}~and~\ref{tab:hw_gao}; thus for these two codes, the conclusions in Section~\ref{sec:dir_hw} apply.

\subsection{Errors-and-Erasures Decoding}
The complexity analysis of RS decoding in Sections~\ref{sec:time}~and~\ref{sec:tran} has assumed errors-only decoding.
We extend our complexity analysis to errors-and-erasures decoding below.

Syndrome-based errors-and-erasures decoding has been well studied, and we adopt the approach in~\cite{Moon05}.
In this approach, first erasure locator polynomial and modified syndrome polynomial are computed.
After the error locator polynomial is found by the key equation solver,
the errata locator polynomial is computed and the error and erasure values are computed by Forney's formula.
This approach is used in both direct and fast implementation.

Syndromeless errors-and-erasures decoding can be carried out in two approaches.
Let us denote the number of erasures as $\nu$ ($0\leq \nu \leq 2t$), and up to $f=\lfloor\frac{2t - \nu}{2}\rfloor$ errors can be corrected given $\nu$ erasures.
As pointed out in~\cite{Shiozaki90,Gao03}, the first approach is to ignore the $\nu$ erased coordinates, thereby transforming the problem into errors-only decoding of an $(n-\nu, k)$ shortened RS code.
This approach is more suitable for direct implementation.
The second approach is similar to syndrome-based errors-and-erasures decoding described above, which uses the erasure locator polynomial~\cite{Shiozaki90}.
In the second approach, only the partial GCD step is affected, while the same fast implementation techniques described in Section~\ref{sec:tran} can be used in the other steps.
Thus, the second approach is more suitable for fast implementation.

We readily extend our complexity analysis for errors-only decoding in Sections~\ref{sec:time}~and~\ref{sec:tran} to errors-and-erasures decoding.
Our conclusions for errors-and-erasures decoding are the same as those for errors-only decoding: Algorithm~\ref{alg:gao} is the most efficient only for very low rate codes; syndrome-based decoding is the most efficient algorithm for high rate codes.
For brevity, we omit the details and interested readers will have no difficulty filling in the details.

\section{Conclusion}\label{sec:con}
We analyze the computational complexities of two syndromeless decoding algorithms for RS codes using both direct implementation and fast implementation, and compare them with their counterparts of syndrome-based decoding.
With either direct or fast implementation, syndromeless algorithms are more efficient than the syndrome-based algorithms only for RS codes with very low rate.
When implemented in hardware, syndrome-based decoders also have lower complexity and higher throughput.
Since RS codes in practice are usually high-rate codes, syndromeless decoding
algorithms are not suitable for these codes.
Our case study also shows that fast implementations can significantly reduce the decoding complexity.
Errors-and-erasures decoding is also investigated although the details are omitted for brevity.

\section*{Acknowledgment}

This work was financed by a grant from the Commonwealth of Pennsylvania, Department of Community and Economic Development, through the Pennsylvania Infrastructure Technology Alliance (PITA).

The authors are grateful to Dr. J\"urgen Gerhard for valuable discussions. The authors would also like to thank the reviewers for their constructive comments, which have resulted in significant improvements in the manuscript.



\end{document}